\documentclass[12pt, reqno]{amsart}
\usepackage{amsmath, amsthm, amscd, amsfonts, amssymb, graphicx, color}
\usepackage{graphicx,color}
\usepackage{setspace,kantlipsum}
\usepackage{amsfonts}
\usepackage{amssymb}
\usepackage{mathrsfs}
\usepackage{amsmath}
\usepackage{amsthm}
\usepackage{latexsym}
\usepackage{graphicx}
\usepackage{xypic}
\usepackage{hyperref}
\usepackage{lineno}
\usepackage{tikz-cd} 
\usepackage{subcaption} 
\usepackage[bottom]{footmisc}

\makeatletter
\@namedef{subjclassname@2020}{%
 \textup{2020} Mathematics Subject Classification}
\makeatother

\usepackage{doc}
\makeatletter

\numberwithin{equation}{section}

\theoremstyle{definition}
\newtheorem{theorem}{Theorem}[section]

\newtheorem{proposition}[theorem]{Proposition}

\newtheorem{remark}[theorem]{Remark}

\def\Tr{{\rm Tr\,}}

\def\C{\mathbb C}

\def\diag{{\mbox{diag}\,}}

\begin{document}
 
\title[Spectral geometric mean and trace characterizations]{Spectral geometric mean and trace characterizations}

\author{Airat~Bikchentaev}
\address{Airat~Bikchentaev}
\address{Lobachevskii Institute of Mathematics and Mechanics, Kazan Federal University, Kazan, 420008, Russia}
\email{Airat.Bikchentaev@kpfu.ru}

\author{Trung~Hoa~Dinh}
\address{Trung~Hoa~Dinh}
\address{Department of Mathematics and Statistics, Troy University, Troy, AL 36082, USA}
\email{thdinh@troy.edu}

\author{Anh~Vu~Le}
\address{Anh~Vu~Le$^{1,2}$}
\address{1. University of Economics and Law, Ho Chi Minh City, Vietnam} 
\address{2. Vietnam National University, Ho Chi Minh City, Vietnam}
\email{vula@uel.edu.vn}

\author{Mohammad~Sal~Moslehian} 
\address{Mohammad~Sal~Moslehian}
\address{Department of Pure Mathematics, Faculty of Mathematical Sciences, Ferdowsi University of Mashhad, P. O. Box 1159, Mashhad 91775, Iran.}
 \email{moslehian@um.ac.ir; moslehian@yahoo.com}


\begin{abstract}
	We use nearly parallel pure states to characterize positive linear functionals $\phi$ on $\mathbb{M}_n$ as positive multiples of the trace if and only if $\phi(A \natural B) \leq \sqrt{\phi(A) \phi(B)}$ for all positive definite matrices $A$ and $B$. Here $A \natural B = (A^{-1} \# B)^{1/2} A (A^{-1} \# B)^{1/2}$ represents the spectral geometric mean. For further clarification, we establish novel characterizations through the inequality $\phi(A \natural B) \leq \phi((A+B)/2)$ for all positive definite matrices $A$ and $B$. We also present a trace inequality related to quantum fidelity that applies to all positive definite matrices, and demonstrate that it does not characterize the trace. 
\end{abstract}

\maketitle 
\section{Introduction}
Let $\mathbb{M}_n$ denote the full matrix algebra of all $n \times n$ complex matrices equipped with the L\"{o}wner order $\ge$, and let $\mathbb{P}_n$ stand for the cone of positive definite matrices in $\mathbb{M}_n$. Tracial functionals are essential in matrix analysis, operator algebras, and quantum information. On a unital operator algebra $\mathscr{A}$, a tracial state is a normalized (meaning $\varphi(I) = 1$) positive linear functional $\varphi: \mathscr{A} \to \mathbb{C}$ satisfying $\varphi(AB) = \varphi(BA)$ for all $ A, B \in \mathscr{A}$. When $ \mathscr{A}=\mathbb{M}_n$, it corresponds to the normalized trace $\tau(A) = \frac{1}{n} \Tr(A)$. It is known that every positive linear functional $\varphi(X)$ on $\mathbb{M}_n$ can be expressed as $\varphi(X)=\Tr(SX)$ for some unique $S\in \mathbb{P}_n$.

Characterizing this unique $\tau$ clarifies matrix symmetries and quantum states. The trace $\tau(A)$ corresponds to the expectation under the maximally mixed state $\rho_* = \frac{I}{n}$, where von Neumann entropy $ S(\rho) = -\Tr(\rho \log \rho) $ peaks, reflecting maximal disorder and unitary invariance. Thus, characterizing $\tau$ also identifies $\rho_*$, bridging matrix theory and quantum symmetry.

It is shown in \cite{BM2} that $\varphi((A^{1/2} B A^{1/2})^{1/2}) \leq \frac{\varphi(A) + \varphi(B)}{2}$ for $A, B \ge 0$ characterizes the trace functional, up to a scalar. This connects to quantum divergences via $\mathrm{d}_b(A,B) = \Tr((A+B)/2) - \Tr((A^{1/2} B A^{1/2})^{1/2})$,which is half the squared Bures--Wasserstein distance; see \cite{BhatiaJainLim2019, spehner1}.

The homogeneous function $\mathcal{H}(A,B) = (A+B)/2 - (A^{1/2} B A^{1/2})^{1/2}$ generalizes the scalar gap $ h(a,b) = (a+b)/2 - \sqrt{ab} $, inspiring divergences like Hellinger and Bhattacharyya types. Recently, Jeong, Kim, and Tam \cite{Jeong2024} used the spectral geometric mean $ A \natural B $ to define a new quantum divergence.   Related constructions of quantum divergences arising from trace inequalities,
monotonicity inequalities, and operator convexity were studied in
\cite{DinhDuNguyenVuong2023,Dinh2013b,DinhLeOsakaPhan2025,DinhTikhonov2010b}. 

In this paper, we introduce a novel rank-one projection test, using orthogonal projections (quantum pure states) to  prove  tracial properties. We prove that either $\varphi(A\natural B) \leq \sqrt{\varphi(A)\varphi(B)}$ or $\varphi(A\natural B) \leq \varphi((A+B)/2)$ characterizes $\varphi = c \Tr$. In addition, a fidelity-based inequality holds but fails to characterize $\Tr$, as shown by counterexamples. 

\section{Preliminaries}

Let $A\# B := A^{1/2}\!\big(A^{-1/2}BA^{-1/2}\big)^{1/2}\!A^{1/2}$ be the (Kubo--Ando) metric geometric mean \cite{KuboAndo1980}, as mentioned in the monograph \cite[Section 3.2]{MO}.
Following Fiedler--Pt\'ak \cite{FiedlerPtak1997}, the \emph{spectral geometric mean} of $A,B\in\mathbb{P}_n$ is defined as:
\begin{equation}\label{eq:sgm-def}
	A \,\natural\, B \;:=\; \big(A^{-1}\# B\big)^{1/2}\; A\; \big(A^{-1}\# B\big)^{1/2}\ \in \mathbb{P}_n.
\end{equation}
For two matrices $A, B \in \mathbb{P}_n$, the metric and spectral geometric means are equal if and only if $A$ and $B$ commute \cite[Theorem 5.2]{FiedlerPtak1997}.

A matrix-valued map $F: \mathbb{P}_n \times \mathbb{P}_n \to \mathbb{P}_n$ is \emph{positively homogeneous of degree 1} if for any $X, Y \in \mathbb{P}_n$ and $t > 0,$ 
\[
F(tX, tY) = t F(X, Y).
\]
If $A$ and $G$ are such maps with $A(X, Y) \ge G(X, Y)$ for all $X, Y$, then for every positive linear functional $\varphi$ (in particular, the trace),
\[
\varphi(A(X, Y)) \geq \varphi(G(X, Y)),
\]
and the difference $A - G$ inherits homogeneity: $(A - G)(tX, tY) = t (A - G)(X, Y)$.
This observation underpins many ``$A \geq G$'' matrix inequalities: homogeneity ensures scale invariance, while the Loewner order $A \ge G$ yields a scalar inequality under $\Tr$.

The following result is well-known in literature. 
\begin{proposition}\label{pro1} 
	Let $\varphi: \mathbb{M}_n \to\mathbb{C}$ be a linear functional. The following are equivalent:
	\begin{enumerate}\itemsep.2em
		\item[(i)] $\varphi$ is tracial, that is,\ $\varphi(AB)=\varphi(BA)$ for all $A,B \in \mathbb{M}_n$;
		\item[(ii)] $\varphi$ is unitarily invariant: $\varphi(U^*AU)=\varphi(A)$ for all unitary $U\in \mathbb{M}_n$ and all $A\in \mathbb{M}_n$;
		\item[(iii)] $\varphi$ takes the same value on every rank-one projection.
	\end{enumerate}
	In this case $\varphi=\alpha\,\Tr$ for some $\alpha\in\mathbb{C}$, and for any projection $P$ of rank $k$,
	\[
	\varphi(P)=\alpha\,\Tr(P)=\frac{k}{n}\,\varphi(I).
	\]
	If, in addition, $\varphi$ is positive, then $\alpha=\varphi(I)/n\geq 0$.
\end{proposition}

\section{Spectral geometric mean and trace characterizations}
\subsection{Rank-one perturbation witness}

Let us start with the inequality:
\begin{equation}\label{eq:target}
	\varphi\Big(\, (A^{1/2} B A^{1/2})^{1/2}\,\Big)
	\leq \frac{\varphi(A)+\varphi(B)}{2}.
\end{equation}
The main goal is to find a counterexample showing the violation of this inequality if 
the functional $\varphi$ is not tracial. When $A$ and $B$ are rank-one 
(or close to rank-one), the algebra simplifies significantly. This naturally leads us to consider using
rank-one projections, which in the context of quantum information theory correspond to 
pure states, as test objects. The strategy involves selecting nearly parallel pure states $u$ and $v$, where the ``transition probability'' $|\langle u,v\rangle|^2$ is close to $1$ \cite{Uhlmann1976}. 
Subsequently, we test the inequality using rank-one matrices
\[
A_0=\lambda\,|u\rangle \langle u|,\qquad B_0=\lambda\,|v\rangle \langle v|\qquad(\lambda>0).
\]
If the inequality fails in this simple setting, it will also fail for nearby choices 
of $A$ and $B$ in higher dimensions. In fact, for general $n\times n$ matrices, one can take
\[
A=\varepsilon I+\lambda |u\rangle \langle u|,\qquad B=\varepsilon I+\lambda |v\rangle \langle v|,
\]
with $\varepsilon>0$ being small and $\lambda$ being large, so that $A$ and $B$ are positive definite and 
close to the rank-one test case $(A_0,B_0)$. Such matrices play a special role in quantum information theory when $\varepsilon + \lambda =1$: they represent noisy pure states (mixed states obtained by adding depolarizing noise to a rank-one projection). All values in the inequality vary continuously based on $(\varepsilon,\lambda)$, maintaining the same sign pattern as in the rank-one case. 
This indicates that testing on rank-one projections is adequate for detecting violations of the inequality.

Choosing $u$ and $v$ to be nearly parallel mimics tasks in quantum hypothesis testing \cite{Audenaert2007}, where one seeks to distinguish between two quantum states that have high overlap. The transition probability $|\langle u,v\rangle|^2$ measures how likely it is that one pure state is confused for another, while the functional values $\varphi(A)$ and $\varphi(B)$ play the role of expected outcomes under a measurement \cite{Hayashi2017}. In this sense, if the functional $\varphi$ is not tracial, then for carefully chosen nearly parallel pure states $u$ and $v$ the inequality fails. Just as in quantum information, a nontrivial measurement can detect distinguishability of close states. This provides a bridge between operator algebraic characterizations of the trace and operational ideas of fidelity and state discrimination in quantum information theory \cite{Jozsa1994}.

Now, let's return to inequality (\ref{eq:target}). We choose unit vectors
\[
u=(\cos\theta,\sin\theta)^\top,\qquad v=(\cos(\theta+\Delta),\sin(\theta+\Delta))^\top
\]
with $|\Delta|\ll1$ so that $|\langle u,v\rangle|=\cos\Delta\approx1-\tfrac12\Delta^2$ is close to $1$.
Let
\[
A_0=\lambda\,|u\rangle \langle u|,\qquad B_0=\lambda\,|v\rangle \langle v|\qquad(\lambda>0).
\]
Then $A_0^{1/2}=\sqrt{\lambda}\,|u\rangle \langle u|$ and
\[
A_0^{1/2}B_0A_0^{1/2}
=\lambda^2(|u\rangle \langle u|)(|v\rangle \langle v|)(|u\rangle \langle u|)
=\lambda^2|\langle u,v\rangle|^2\,|u\rangle \langle u|,
\]
whence
\[
\big(A_0^{1/2}B_0A_0^{1/2}\big)^{1/2}
=\lambda\,|\langle u,v\rangle|\,|u\rangle \langle u|.
\]
Considering $\varphi(X)=\operatorname{Tr}(SX)$ with $S\ge0$, we obtain
\[
\varphi \Big(\big(A_0^{1/2}B_0A_0^{1/2}\big)^{1/2}\Big)
=\lambda\,|\langle u,v\rangle|\,\langle u,Su\rangle,~~
\varphi(A_0)=\lambda\,\langle u,Su\rangle,~~ \mbox{and}~~
\varphi(B_0)=\lambda\,\langle v,Sv\rangle.
\]

 By cancelling the common factor $\lambda$ in (\ref{eq:target}), we see that
(\ref{eq:target}) for $(A_0,B_0)$ is equivalent to
\begin{equation}\label{eq:F-again}
	F(\theta,\Delta):=
	2\,|\langle u,v\rangle|\,\langle u,Su\rangle
	-\langle u,Su\rangle-\langle v,Sv\rangle \; \le \; 0.
\end{equation}
Thus, a violation of (\ref{eq:target}) is equivalent to $F(\theta,\Delta)>0$.

Since (\ref{eq:target}) is homogeneous in $\varphi$, we may normalize $S$ so that
its smaller eigenvalue equals $1$. In a suitable orthonormal basis,
we write
\[
S=\operatorname{diag}(s,1),\qquad s>1 .
\]
Then
\[
\langle u,Su\rangle=1+(s-1)\cos^2\theta,\qquad
\langle v,Sv\rangle=1+(s-1)\cos^2(\theta+\Delta),
\]
and $|\langle u,v\rangle|=\cos\Delta$ for $|\Delta|$ small. Using
\[
\cos\Delta=1-\frac12\Delta^2+O(\Delta^4),
\]
and
\[
\cos^2(\theta+\Delta)
=
\cos^2\theta-\sin(2\theta)\Delta-\cos(2\theta)\Delta^2+O(\Delta^3),
\]
we obtain
\[
F(\theta,\Delta)
=
(s-1)\sin(2\theta)\Delta
-
\bigl(1+(s-1)\sin^2\theta\bigr)\Delta^2
+O(\Delta^3).
\]
If $\theta\in(0,\pi/2)$ and $\Delta>0$ is sufficiently small, then
the linear term is positive and dominates the higher-order terms.
Hence $F(\theta,\Delta)>0$, so (\ref{eq:target}) fails.

Equally, one may take $\theta\in(\pi/2,\pi)$ and $\Delta<0$ sufficiently
small. This gives the same conclusion. 
  
   To obtain positive definite test matrices, we regularize the rank-one
construction. Namely, embed the above two-dimensional construction into
$\mathbb C^n$ and set
\[
A_\varepsilon=\varepsilon I+\lambda |u\rangle \langle u|,\qquad
B_\varepsilon=\varepsilon I+\lambda |v\rangle \langle v|,
\qquad \varepsilon>0.
\]
The parameter $\varepsilon$ is introduced only to make the rank-one matrices
strictly positive, so that $A_\varepsilon,B_\varepsilon\in\mathbb P_n$.
Since the violation obtained above at $(A_0,B_0)$ is strict and the map
\[
(A,B)\mapsto
\varphi\!\left((A^{1/2}BA^{1/2})^{1/2}\right)
-\frac{\varphi(A)+\varphi(B)}2
\]
is continuous on the positive semidefinite cone, the strict violation
persists for all sufficiently small $\varepsilon>0$. Hence there exist
$A,B\in\mathbb P_n$ for which \eqref{eq:target} fails. 

There are numerous characterizations of trace functionals on matrices and operator algebras; see the survey article \cite{BM3} as well as \cite{BM1} and references therein. We have the following characterization that was recently obtained in \cite{BM2}. Here, we reformulate it in a slightly different form.
\begin{proposition} 
	Let $\varphi(X)=\Tr(SX)$ be a positive linear functional on $\mathbb{M}_n$ with
	$S\in \mathbb{P}_n$ not being a scalar multiple of the identity. Then, there exist
	$A,B\in \mathbb{P}_n$ such that
	\begin{equation*}\label{eq:target}
		\varphi\Big(\, (A^{1/2} B A^{1/2})^{1/2}\,\Big)
		\;>\; \frac{\varphi(A)+\varphi(B)}{2}.
	\end{equation*}
\end{proposition}

\subsection{The spectral geometric mean and characterization of trace}

Recall from \cite[Proposition 3.2.9]{BKMS} that $X=A^{-1}\#B$ if and only if $B=XAX$. With such a matrix $X$ we have:
\begin{equation}\label{eq:spec-def}
	A\natural B \;=\; X^{1/2} A X^{1/2}.
\end{equation}

By Cauchy--Schwarz's inequality in the Hilbert--Schmidt inner product, we can see that
\begin{align*}
	\Tr(A\natural B)=\Tr(X^{1/2}AX^{1/2})
	=\langle A^{1/2},\,X A^{1/2}\rangle_{HS}
	&\leq \|A^{1/2}\|_{HS}\,\|X A^{1/2}\|_{HS} \\
	&=\sqrt{\Tr A\;\Tr(XAX)} =\sqrt{\Tr A\;\Tr B}.
\end{align*}

\begin{theorem} 
	Let $\varphi(X)=\Tr(SX)$ be a positive linear functional on $\mathbb{M}_n$ with $S\in \mathbb{P}_n$.
	Assume that for all $A,B\in \mathbb{P}_n$,
	\begin{equation}\label{eq:specGM}
		\varphi(A \natural B)\ \le\ \sqrt{\varphi(A)\,\varphi(B)} .
	\end{equation}
	Then, $S$ is a scalar multiple of the identity, that is,\ $\varphi=c\,\Tr$ with $c\ge0$.
\end{theorem}

\begin{proof}
	Based on (\ref{eq:spec-def}) we can conclude that \eqref{eq:specGM} is equivalent to
	\begin{equation}\label{eq:star}
		\Tr\!\big(S\,X^{1/2}AX^{1/2}\big)\ \le\ 
		\sqrt{\,\Tr(SA)\,\Tr(S\,XAX)},\qquad \mbox{for all~} A,X\in \mathbb{P}_n.
	\end{equation}
	If $S$ is not a scalar, we can perform a unitary change of basis and focus on the $2\times2$ corner
	to assume $S=\diag(s,1)$ with $s\neq1$.
	Let us fix the following unit vectors:
	\[
	u=\binom{\cos\theta}{\sin\theta},\quad
	v=\binom{\cos(\theta+\Delta)}{\sin(\theta+\Delta)},
	\]
	and parameters $\lambda,\mu>0$.   For $\varepsilon>0$ set
\[
A_\varepsilon:=\varepsilon I+\lambda\,|u\rangle \langle u|,\quad
X_\varepsilon:=\varepsilon I+\mu\,|v\rangle \langle v|,\quad \mbox{and} \quad
B_\varepsilon:=X_\varepsilon A_\varepsilon X_\varepsilon .
\]
Then $A_\varepsilon,B_\varepsilon\in \mathbb{P}_n$ and, in view of \eqref{eq:spec-def},
\[
A_\varepsilon\natural B_\varepsilon
=
X_\varepsilon^{1/2}A_\varepsilon X_\varepsilon^{1/2}.
\]
Let $P_u:=|u\rangle\langle u|$ and $P_v:=|v\rangle\langle v|$. Since
\[
X_\varepsilon
=
(\varepsilon+\mu)P_v+\varepsilon(I-P_v),
\]
we have
\[
X_\varepsilon^{1/2}
=
(\varepsilon+\mu)^{1/2}P_v+\varepsilon^{1/2}(I-P_v)
=
\mu^{1/2}P_v+O(\varepsilon^{1/2})
\]
in operator norm, which means that the norm of the error is bounded by a constant times $\varepsilon^{1/2}$. Moreover,
\[
A_\varepsilon=\lambda P_u+O(\varepsilon).
\]
Hence
\[
X_\varepsilon^{1/2}A_\varepsilon X_\varepsilon^{1/2}
=
\lambda\mu P_vP_uP_v+O(\varepsilon^{1/2})
=
\lambda\mu|\langle u,v\rangle|^2P_v+O(\varepsilon^{1/2}).
\]
Therefore, as $\varepsilon\downarrow0$,
\begin{align*}
\Tr\!\big(S\,X_\varepsilon^{1/2}A_\varepsilon X_\varepsilon^{1/2}\big)
&=
\lambda\mu\,|\langle u,v\rangle|^{2}\,\langle v,Sv\rangle
+O(\varepsilon^{1/2}),\\
\Tr(SA_\varepsilon)
&=
\lambda\,\langle u,Su\rangle+O(\varepsilon),\\
\Tr\!\big(S\,X_\varepsilon A_\varepsilon X_\varepsilon\big)
&=
\lambda\mu^{2}\,|\langle u,v\rangle|^{2}\,\langle v,Sv\rangle
+O(\varepsilon).
\end{align*} 
	Plugging these into \eqref{eq:star}, canceling the common positive factor
	 $\lambda \mu|\langle u,v\rangle|$ and letting $\varepsilon\to0$ we obtain the following inequality
	\begin{equation}\label{eq:core}
		|\langle u,v\rangle|\ \le\ \sqrt{\frac{\langle u,Su\rangle}{\langle v,Sv\rangle}}
		\qquad (\text{for all unit }u,v).
	\end{equation}
	Put
	\[
	f(\theta):=\langle u,Su\rangle
	= 1+(s-1)\cos^2\theta.
	\]
	The inequality \eqref{eq:core} becomes 
	\begin{equation}\label{eq:angle}
		|\cos\Delta| \le\ \sqrt{\frac{f(\theta)}{f(\theta+\Delta)}},
	\end{equation}
	or equivalently,
	\[
	\cos^{2}\Delta\; f(\theta+\Delta) \ \le\ f(\theta).
	\]
	Expanding for small $\Delta$ as
	\[
	\cos^{2}\Delta \;=\; 1-\Delta^{2}+O(\Delta^{4}),\qquad
	f(\theta+\Delta)=f(\theta)+f'(\theta)\Delta+O(\Delta^{2}),
	\]
	with $f'(\theta)=-(s-1)\sin(2\theta)$, we get 
	\[
	\cos^{2}\Delta\; f(\theta+\Delta)
	= f(\theta)\ +\ f'(\theta)\Delta\ -\ f(\theta)\Delta^{2}\ +\ O(\Delta^{2}).
	\]
	If $s\neq1$, choose $\theta$ so that $\sin(2\theta)\neq0$, and take the sign of $\Delta$
	so that $f'(\theta)\Delta>0$. For sufficiently small $|\Delta|>0$ the linear term then dominates
	the $O(\Delta^{2})$ terms and we obtain
	\[
	\cos^{2}\Delta\; f(\theta+\Delta)\ >\ f(\theta),
	\]
	contradicting \eqref{eq:core}. Thus no such $s\neq1$ can exist, and in every $2\times2$ corner
	$S$ must be a scalar; hence $S=cI$. 
\end{proof}
 
 \begin{theorem}
	A positive linear functional $\varphi$ on $\mathbb{M}_n$ satisfies
	\begin{equation}\label{eq:specGM2}
		\varphi\big((A \natural B)^2\big)
		\le
		\sqrt{\varphi(A^2)\,\varphi(B^2)}
	\end{equation}
	for all $A,B\in \mathbb{P}_n$ if and only if
	$\varphi$ is a nonnegative scalar multiple of the trace.
\end{theorem} 
\begin{proof}
	It follows from \cite[Theorem 3.2(8)]{FiedlerPtak1997} that the eigenvalues of $A \natural B$ are the positive square roots of the eigenvalues of 
	$AB$. It follows from spectral theorem that
	\begin{align*}
		\Tr\big((A \natural B)^2\big)&=\sum_{j=1}^n\lambda_j\big((A \natural B)^2\big)=\sum_{j=1}^n\lambda_j^2(A \natural B)=\sum_{j=1}^n\lambda_j(AB)\\
		&=\Tr(AB) \leq \sum_{j=1}^n\lambda_j(A)\lambda_j(B)\\
		&\qquad \qquad \text{(by~the~von~Neumann~trace~inequality~ \cite[Theorem~2.1.25]{BKMS})}\\
		&\leq \left(\sum_{j=1}^n\lambda_j(A)^2\right)^{1/2} \left(\sum_{j=1}^n\lambda_j(B)^2\right)^{1/2}\\
		&=\sqrt{\Tr(A^2)\,\Tr(B^2)}\,.
	\end{align*} 
    Therefore, if $\varphi=c\Tr$ with $c\ge0$, then
\[
\varphi\big((A\natural B)^2\big)
=
c\,\Tr\big((A\natural B)^2\big)
\le
c\sqrt{\Tr(A^2)\Tr(B^2)}
=
\sqrt{\varphi(A^2)\varphi(B^2)}.
\]
Thus every nonnegative scalar multiple of the trace satisfies
\eqref{eq:specGM2}. Conversely, let $\varphi(X)=\Tr(SX)$ be a positive linear functional on
$\mathbb{M}_n$ satisfying \eqref{eq:specGM2}, where $S\ge0$. If $S=0$,
then $\varphi=0$ is a nonnegative scalar multiple of the trace. Assume
therefore that $S\ne0$.

It is enough to show that the restriction of $S$ to every two-dimensional
subspace is scalar. Suppose, to the contrary, that $S$ has two
distinct eigenvalues. Restricting to the corresponding two-dimensional
subspace and normalizing the restricted functional, we may assume that
$n=2$, $\Tr(S)=1$, and
\[
S=\operatorname{diag}\left(\frac12+s,\frac12-s\right),
\qquad 0<s\le \frac12 .
\]
We shall derive a contradiction.

Consider the following matrices, where $\varepsilon>0$ is a real number:
	$$
	A = \frac12 \begin{pmatrix} 1 & 1 \\ 1 & 1+\varepsilon \end{pmatrix} \quad \mbox{and}\quad X = \frac14\begin{pmatrix} 2+2\varepsilon+\varepsilon^2 & 2+\varepsilon \\ 2+\varepsilon & 2 \end{pmatrix}.
	$$
	Then, 
	$$
	X^{1/2} = \frac12\begin{pmatrix} 1+\varepsilon & 1 \\ 1 & 1 \end{pmatrix}
	$$
	and
	$$
	XAX = \frac{1}{32}\begin{pmatrix} 16+28\varepsilon+21\varepsilon^2+7\varepsilon^3+\varepsilon^4 & 16+20\varepsilon+9\varepsilon^2+\varepsilon^3 \\ 16+20\varepsilon+9\varepsilon^2+\varepsilon^3 & 16+12\varepsilon+\varepsilon^2 \end{pmatrix}.
	$$
	Set $B:=XAX$. We have 
	$$
	A\natural B=X^{1/2}AX^{1/2}=\frac{1}{8}\begin{pmatrix} 4+5\varepsilon+\varepsilon^2 & 4+3\varepsilon \\ 4+3\varepsilon & 4+\varepsilon \end{pmatrix}
	,
	$$
	$$
	(A\natural B)^2=\frac{1}{64}\begin{pmatrix} 32+64\varepsilon+42\varepsilon^2+10\varepsilon^3+\varepsilon^4 & 
		32+48\varepsilon+22\varepsilon^2+3\varepsilon^3 
		\\ 32+48\varepsilon+22\varepsilon^2+3\varepsilon^3 & 32+32\varepsilon+10 \varepsilon^2 \end{pmatrix}.
	$$ 
	Therefore, 
	\begin{equation}
		\begin{split}
			\varphi ( (A\natural B)^2) & =
			\frac{1}{64}(32+64\varepsilon+42\varepsilon^2+10\varepsilon^3+\varepsilon^4)\left(\frac12 +s\right)
			+ \frac{1}{64}(32+32\varepsilon+10\varepsilon^2)\left(\frac12 -s\right)
			\\ & = \frac12 +\frac34 \varepsilon +\frac{s}{2}\varepsilon +o(\varepsilon) \quad \mbox{as}\quad \varepsilon\to 0^+.
		\end{split}
		\notag
	\end{equation}
	Since 
	$$
	B= \frac{1}{32}\begin{pmatrix} 16+28\varepsilon+o(\varepsilon) & 16+20\varepsilon+o(\varepsilon) \\ 16+20\varepsilon+o(\varepsilon) & 16+12\varepsilon+o(\varepsilon) \end{pmatrix}
	\quad \mbox{as}\quad \varepsilon\to 0^+,
	$$
	we obtain 
	$$
	B^2= \frac{1}{2}\begin{pmatrix} 1+3\varepsilon+o(\varepsilon) & 1+\frac52 \varepsilon+o(\varepsilon) \\ 1+\frac52 \varepsilon+o(\varepsilon) & 1+2\varepsilon+o(\varepsilon) \end{pmatrix}
	\quad \mbox{as}\quad \varepsilon\to 0^+,
	$$
	and hence, 
	\begin{equation}
		\begin{split}
			\varphi ( B^2) & =
			\frac{1}{2}\left( (1+3\varepsilon)\left(\frac12 +s\right)
			+ (1+2\varepsilon)\left(\frac12 -s\right)\right) +o(\varepsilon)
			\\ & = \frac12 +\frac54 \varepsilon +\frac{s}{2}\varepsilon +o(\varepsilon) \quad \mbox{as}\quad \varepsilon\to 0^+.
		\end{split}
		\notag
	\end{equation}
	Since
	$$
	A^{2} = \frac14\begin{pmatrix} 2 & 2+\varepsilon \\ 2+\varepsilon & 2+2\varepsilon+\varepsilon^2 \end{pmatrix},
	$$
	we have
	$$
	\varphi ( A^2) =
	\frac{1}{2}+\frac{\varepsilon}{4}- \frac{s}{2}\varepsilon +o(\varepsilon) \quad \mbox{as}\quad \varepsilon\to 0^+.
	$$
	By squaring both sides of inequality \eqref{eq:specGM2}, we arrive at
	$$
	\left( \frac12 +\frac34 \varepsilon +\frac{s}{2}\varepsilon +o(\varepsilon) \right)^2\leq
	\left(\frac{1}{2}+\frac{\varepsilon}{4}- \frac{s}{2}\varepsilon +o(\varepsilon)\right)
	\left(\frac{1}{2}+\frac{5}{4}\varepsilon+ \frac{s}{2}\varepsilon +o(\varepsilon)\right) 
	\quad \mbox{as}\quad \varepsilon\to 0^+,
	$$
	that is, 
	$$
	\frac14 +\frac34 \varepsilon +\frac{s}{2}\varepsilon \leq \frac14 +\frac34 \varepsilon +o(\varepsilon)
	\quad \mbox{as}\quad \varepsilon\to 0^+.
	$$
	Thus,
	$$
	\frac{s}{2}\varepsilon \leq o(\varepsilon) \quad \mbox{as}\quad \varepsilon\to 0^+.
	$$
	Dividing both sides of this inequality by the positive number $\varepsilon$ and then passing to the limit as
	$ \varepsilon\to 0^+$ we reach $s\leq 0$.   This contradicts $s>0$. Hence no two distinct eigenvalues of $S$ can
exist. Therefore $S=cI$ for some $c\ge0$, and so
$\varphi=c\Tr$. 
	
\end{proof}

 \begin{remark}
When $B=I$ and $\varphi$ is a state, inequality (\ref{eq:specGM2})
formally reduces to the Kadison inequality
\[
\varphi(A)^2 \le \varphi(A^2);
\]
see \cite[Theorem 5.2.6]{BKMS} and \cite{KMN}. However, (\ref{eq:specGM2}) should not be regarded as a direct extension of
Kadison's inequality. Indeed, the latter holds for every state on $M_n$,
whereas (\ref{eq:specGM2}), when required to hold for all $A,B\in\mathbb{P}_n$,
forces $\varphi$ to be a scalar multiple of the trace. Thus, (\ref{eq:specGM2}) is better viewed as a two-variable strengthening
of Kadison's inequality that characterizes the trace.
\end{remark}

Although for any positive linear map, the inequality $$\varphi(A\# B) \leq \varphi \left(\frac{A+B}{2}\right)$$ follows from the matrix arithmetic-geometric mean inequality $A\# B \leq (A+B)/2$ for any $A, B \succ 0$, however, $\varphi(A\natural B) \leq \varphi (\frac{A+B}{2})$ does not hold for a positive linear functional on $\mathbb{M}_n$, in general. A counterexample is as follows: Consider the matrices 
\[
A = \begin{pmatrix} 4 & 0 \\ 0 & 1 \end{pmatrix} \quad 
B = \begin{pmatrix} 1 & 0 \\ 0 & 1 \end{pmatrix}, \quad
\]
for which $A \natural B =\begin{pmatrix} 2 & 0 \\ 0 & 1 \end{pmatrix}$ and the positive linear functional $\varphi(X)=\Tr(SX)$, where $S = \begin{pmatrix} 1 & 0 \\ 0 & 2 \end{pmatrix}$. However, note that 
$$
\Tr(A\natural B) = \Tr((A^{1/2}BA^{1/2})^{1/2}) \leq \Tr\left (\frac{A+B}{2}\right). 
$$

\begin{theorem} 
	Let $\varphi(X)=\Tr(SX)$ be a positive linear functional on $\mathbb M_n$ with
	$S\in \mathbb{P}_n$. 
	If
	\begin{equation}\label{eq:main}
		\varphi(A\natural B)\ \le\ \varphi\Big(\frac{A+B}{2}\Big),\qquad \text{for all~}A,B\in \mathbb{P}_n,
	\end{equation}
	then $S$ is a scalar multiple of the identity. 
\end{theorem}

\begin{proof} 
	Assume \eqref{eq:main} holds for all $A,B\in \mathbb{P}_n$. We show that $S$ must be a multiple of $I$.
	Let us fix unit vectors $u,v\in\C^n$ and $\lambda,\mu>0$, and consider the rank-one
	regularizations
	\[
	A_\varepsilon:=\varepsilon I+\lambda\,|u\rangle\!\langle u|,\quad
	X_\varepsilon:=\varepsilon I+\mu\,|v\rangle\!\langle v|,\quad \mbox{and} \quad
	B_\varepsilon:=X_\varepsilon A_\varepsilon X_\varepsilon,
	\]
	and the rank-one projections $P_u:=|u\rangle\langle u|$ and $P_v:=|v\rangle\langle v|$. Clearly $A_\varepsilon,B_\varepsilon\in \mathbb{P}_n$.    As shown in the proof of Theorem 3.2, we have 
\[
A_\varepsilon\natural B_\varepsilon
=
X_\varepsilon^{1/2}A_\varepsilon X_\varepsilon^{1/2}
=
\lambda\mu|\langle u,v\rangle|^2P_v+O(\varepsilon^{1/2}).
\]
Also,
\[
A_\varepsilon=\lambda P_u+O(\varepsilon),
\]
and
\[
B_\varepsilon=X_\varepsilon A_\varepsilon X_\varepsilon
=
\lambda\mu^2P_vP_uP_v+O(\varepsilon)
=
\lambda\mu^2|\langle u,v\rangle|^2P_v+O(\varepsilon).
\]
Therefore, as $\varepsilon\downarrow0$,
\[
\Tr\!\big(S(A_\varepsilon\natural B_\varepsilon)\big)
=
\lambda\mu\,\alpha\,\beta+O(\varepsilon^{1/2}),
\]
\[
\Tr(SA_\varepsilon)
=
\lambda\,\gamma+O(\varepsilon),
\]
and
\[
\Tr(SB_\varepsilon)
=
\lambda\mu^{2}\,\alpha\,\beta+O(\varepsilon),
\]
where $\alpha:=|\langle u,v\rangle|^{2}\in[0,1],$ $\beta:=\langle v,Sv\rangle,$ $
\gamma:=\langle u,Su\rangle.$ 
	Applying \eqref{eq:main} to $(A_\varepsilon,B_\varepsilon)$ and letting
	$\varepsilon\downarrow0$ yields that
	\[
	\mu\,\alpha\,\beta\ \le\ \frac12\big(\gamma+\mu^{2}\alpha\beta\big)
	\]
	for all $\mu>0$.  Equivalently,
\[
0\le \frac12\gamma+\frac12\mu^2\alpha\beta-\mu\alpha\beta,
\qquad \mu>0.
\]
  Minimizing  the right-hand side with
respect to $\mu>0$, whose minimum is attained at $\mu=1$, we obtain
\[
0\le \frac12\gamma-\frac12\alpha\beta.
\]
Hence
\begin{equation}\label{eq:key}
		{\ \gamma\ \ge\ \alpha\,\beta\ }\qquad\text{for all unit vectors }u,v.
	\end{equation} 
	
	We claim that \eqref{eq:key} can hold for all $u,v$ only if $S$ is a scalar.
	Indeed, if we restrict to any $2$–dimensional subspace on which (in some orthonormal basis)
	$
	S|_{\mathrm{span}\{e_1,e_2\}}=\diag(a,b)
	$
	with $a\neq b$; we will derive a contradiction.

	Let $z:=v=(\cos\phi,\sin\phi)$ where $\phi\in(0,\pi/2)$ (so $v$ is not an
	eigenvector of $S$ on this plane). Then $\beta=z^\ast Dz$ with $D:=\diag(a,b)$.
	Let $x:=u$ and observe that \eqref{eq:key} on this plane is
	 
\[ x^*Dx \ge |z^*x|^2(z^*Dz), \quad \text{for all unit }x.     	
\]     Since $|z^*x|^2=x^*zz^*x$, this is equivalent to   
	\[
	x^\ast\!\big(D-\beta\,zz^\ast\big)x\ \ge\ 0,\qquad\text{for all unit }x.
	\]
	This ensures that the $2\times2$ matrix $M:=D-\beta zz^\ast$ is positive semidefinite.  Indeed, since $D>0$, we can factor
\[
M=D-\beta zz^*
=
D^{1/2}\left(I-\beta D^{-1/2}zz^*D^{-1/2}\right)D^{1/2}.
\]
Hence $M \ge 0$ if and only if
\[
I-\beta D^{-1/2}zz^*D^{-1/2}\ge0.
\]
Since $D^{-1/2}zz^*D^{-1/2}$ is a rank-one positive semidefinite matrix,
which is equivalent to
\[
\beta\,\|D^{-1/2}z\|^2\le1,
\]
that is,
\[
\beta z^*D^{-1}z\le1.
\]  
	 Here $\beta=z^*Dz$. Since $D>0$, we have
\[
1=z^*z=\langle D^{1/2}z,D^{-1/2}z\rangle .
\]
Hence, by the ordinary Cauchy--Schwarz inequality,
\[
1
\le
\|D^{1/2}z\|^2\,\|D^{-1/2}z\|^2
=
(z^*Dz)(z^*D^{-1}z).
\]
Moreover, equality holds if and only if $D^{1/2}z$ and $D^{-1/2}z$
are linearly dependent. This is equivalent to $Dz$ being a scalar
multiple of $z$, that is, $z$ is an eigenvector of $D$.  Here $\beta=z^*Dz$. Since $D=\operatorname{diag}(a,b)$ and
$z=(\cos\phi,\sin\phi)$ with $\phi\in(0,\pi/2)$, we have
\[
\beta=a\cos^2\phi+b\sin^2\phi
\]
and
\[
z^*D^{-1}z=\frac{\cos^2\phi}{a}+\frac{\sin^2\phi}{b}.
\]
Therefore
\[
\begin{aligned}
\beta z^*D^{-1}z
&=(a\cos^2\phi+b\sin^2\phi)
\left(\frac{\cos^2\phi}{a}+\frac{\sin^2\phi}{b}\right)\\
&=\cos^4\phi+\sin^4\phi
+\left(\frac{a}{b}+\frac{b}{a}\right)\cos^2\phi\sin^2\phi\\
&=1+\left(\frac{a}{b}+\frac{b}{a}-2\right)\cos^2\phi\sin^2\phi\\
&=1+\frac{(a-b)^2}{ab}\cos^2\phi\sin^2\phi.
\end{aligned}
\]
Since $a\ne b$ and $\phi\in(0,\pi/2)$, the last term is strictly positive.
Thus
\[
\beta z^*D^{-1}z>1.
\] Therefore the necessary condition for $M\ge0$ fails. Hence $M$ is not
positive semidefinite, and so there exists a unit vector  $x=u$ with $x^\ast Mx<0$, that is,\ $\gamma-\alpha\beta<0$,
	contradicting \eqref{eq:key}.
	
	We conclude that on every $2$–dimensional subspace, the restriction of $S$
	is a scalar multiple of the identity. Therefore, $S=cI$ on $\C^n$.
\end{proof} 

\begin{remark} 
	The scalar $\alpha=|\langle u,v\rangle|^{2}$ is the \emph{fidelity} of the two pure states
	$|u\rangle,|v\rangle$, while $\gamma=\Pr_{S}(u)$ and $\beta=\Pr_{S}(v)$ are Born probabilities in
	state $S$. The necessary and sufficient rank-one consequence \eqref{eq:key} states that
	\[
	 {\Pr}_{S}(u)\ \ge\ F(u,v)\,{\Pr}_{S}(v)\quad\text{for all pure states }u,v.
	\]
	No non-maximally mixed state can satisfy this for all directions; the only states that do are the maximally mixed ones $S = cI$.
\end{remark}

\section{Quantum fidelity and trace characterizations}

For density operators $\rho,\sigma\geq 0$ with $\Tr\rho=\Tr\sigma=1$, we can define
\[
\mathcal{F}(\rho,\sigma)\ :=\ \Tr\!\sqrt{\sqrt{\rho}\,\sigma\,\sqrt{\rho}},
\qquad
F(\rho,\sigma)\ :=\ \bigl(\mathcal{F}(\rho,\sigma)\bigr)^2.
\]
We call $\mathcal{F}$ the \emph{fidelity amplitude} and $F$ the \emph{(Uhlmann) fidelity}.

The relationship between quantum overlap $\Tr(\rho\,\sigma)$ and quantum fidelity is as follows: For all states $\rho$ and $\sigma$,
\begin{equation}\label{eq:overlap-fid}
	\Tr(\rho\,\sigma)\ \le\ F(\rho,\sigma).
\end{equation}

In this section, we discuss when the inequality (\ref{eq:overlap-fid}) could characterize the trace.

\begin{proposition}
	Let $S \ge 0$ be a density matrix on $\mathbb C^n$ (that is,\ $\Tr S=1$) with $n\ge2$. Then, there is no such $S$ for which
	\begin{equation}\label{eq:quad-square}
		\Tr(SY^{2})\ \le\ \bigl(\Tr(SY)\bigr)^{2}\qquad\text{for all }Y\in \mathbb{P}_n .
	\end{equation}
	Equivalently, inequality \eqref{eq:quad-square} cannot hold uniformly over all $Y\in \mathbb{P}_n$
	for any density $S$ in a dimension greater than or equal $2$.
\end{proposition}

\begin{proof}
	First, choose a unit vector $w$ and take $Y=t\,P_w$ where $P_w=|w\rangle\!\langle w|$ and $t>0$.
	Then, we have $\Tr(SY)=t\,\langle w,Sw\rangle$ and $\Tr(SY^{2})=t^{2}\langle w,Sw\rangle$.
	Thus \eqref{eq:quad-square} yields that
	$$
	t^{2}\,\langle w,Sw\rangle\ \le\ t^{2}\,\langle w,Sw\rangle^{2}\qquad \text{or} \qquad \langle w,Sw\rangle\ \le\ \langle w,Sw\rangle^{2}.$$
	This implies that for every unit vector $w$, either $\langle w,Sw\rangle=0$ or $\langle w,Sw\rangle\ge1$.
	In particular, the eigenvalues of $S$ must be in the set $\{0\}\cup[1,\infty)$.
	Since $\Tr S=1$, the only possibility (in any dimension) is that $S$ is a rank-one projection.

	Assume $S=|v\rangle\!\langle v|$ and choose an orthonormal basis with $v=e_1$.
	Let
	\[
	Y=\begin{pmatrix}1&t\\[2pt] t&1\end{pmatrix}\oplus 0_{n-2}\qquad(0<t\le1),
	\]
	which is positive semidefinite (its $2\times2$ principal block has eigenvalues $1\pm t\ge0$).
	Then
	\[
	\Tr(SY)\,=\,\langle e_1,Ye_1\rangle\,=\,1,\qquad
	\Tr(SY^{2})\,=\,\langle e_1,Y^{2}e_1\rangle\,=\,1+t^{2}.
	\]
	Therefore, \eqref{eq:quad-square} would require $1+t^{2}\leq 1$, which is impossible for $t>0$.
  This means that even for a pure state $S$, the inequality fails for some
positive semidefinite $Y$.  
	  Strictly speaking, the above choices of $Y$ are positive semidefinite.
Replacing them by $Y+\delta I$ and letting $\delta\downarrow0$, the same
contradiction follows by continuity. 
	Thus, no density matrix $S$ (with $n\ge2$) can satisfy \eqref{eq:quad-square}
	for all $Y\in \mathbb{P}_n$.
\end{proof}

  \begin{remark}
The inequality
\[
\Tr(SY^2)\le \bigl(\Tr(SY)\bigr)^2,\qquad Y\in\mathbb P_n,
\]
is not homogeneous in $S$. Indeed, replacing $S$ by $cS$ multiplies
the left-hand side by $c$, while the right-hand side is multiplied by
$c^2$. Thus, without a normalization on $S$, this inequality is not
suitable as a trace-characterizing condition.

For instance, in $\mathbb M_2$, if $S=\operatorname{diag}(s,1)$ with
$s\ge1$, then
\[
\Tr(SY^2)\le \bigl(\Tr(SY)\bigr)^2,\qquad Y\in\mathbb P_2.
\]
Thus such estimates may hold for non-normalized non-scalar $S$, for example
for many $S \ge I$. On the other hand, the density constraint
$\Tr S=1$ makes \eqref{eq:quad-square} too strong to hold universally in
dimension $n\ge2$. This shows that a normalization condition on $S$ is
essential if one wants to use this type of inequality to characterize the
trace.
\end{remark} 

However, we have the following characterization.

\begin{proposition}
	Let $\varphi(X)=\Tr(SX)$ with $S\in \mathbb{P}_n$ and $\Tr S=n\ge2$. Then
	\[
	\Tr(SY^{2})\ \le\ \bigl(\Tr(SY)\bigr)^{2}\qquad\text{for all }Y\in \mathbb{P}_n
	\]
	holds if and only if $S=I$ (equivalently, $\varphi=\Tr$).
\end{proposition}

\begin{proof}
	Suppose the inequality holds for all $Y\in \mathbb{P}_n$.
	Let $\lambda_{\min}$ be the smallest eigenvalue of $S$.
	Since $\Tr S=n$, either $S=I$ or $\lambda_{\min}<1$.
	If $\lambda_{\min}<1$, choose a unit eigenvector $w$ with $Sw=\lambda_{\min}w$
	and set $Y:=|w\rangle\langle w|$ (so $Y^2=Y$). Then
	\[
	\Tr(SY^2)=\Tr(SY)=\langle w,Sw\rangle=\lambda_{\min}\in(0,1),
	\]
	which contradicts $\Tr(SY^2)\le(\Tr(SY))^2$ (it would say $x\leq x^2$ with $0<x<1$).
	Hence $\lambda_{\min}\ge1$, and since $\Tr S=n$, all eigenvalues equal $1$:
	$S=I$.
\end{proof}

\medskip
\noindent \textit{Author Contributions Statement.} All authors wrote the main manuscript text and they edited and reviewed the manuscript.

\medskip
\noindent \textit{Conflict of Interest Statement.} On behalf of the authors, the corresponding author states that there is no conflict of interest. 

\medskip
\noindent\textit{Data Availability Statement.} Data sharing is not applicable to this article as no datasets were generated or analysed during the current study.

\medskip
\noindent \textit{Funding Declaration.} { {This research is funded by Vietnam National University Ho Chi Minh City
(VNU-HCM) under grant number B2026-34-06.}}


\begin{thebibliography}{99}
	
	\bibitem{Audenaert2007} K. M. R. Audenaert, J. Calsamiglia, R. Mu\~{n}oz-Tapia, E. Bagan, L. Masanes,
	A. Acín, and F. Verstraete, \textit{Discriminating states: the quantum Chernoff bound}, Phys. Rev. Lett. \textbf{98} (2007), 160501.
	
	\bibitem{BhatiaJainLim2019} R. Bhatia, T. Jain, and Y. Lim, \textit{On the Bures--Wasserstein distance between positive definite matrices}, Expo. Math. \textbf{37} (2019), 165--191.
	
	\bibitem{BKMS} A. M. Bikchentaev, F. Kittaneh, M. S. Moslehian, and Y. Seo, \textit{Trace Inequalities for Matrices and Hilbert Space Operators}, Forum for Interdisciplinary Mathematics, Springer, 2024.
	
	\bibitem{BM1} A. M. Bikchentaev and M. S. Moslehian, \textit{Characterizations of tracial functionals on $C^*$-algebras}, Infin. Dimens. Anal. Quantum Probab. Relat. Top. (2025) 2550003, 14 p.
	
	\bibitem{BM3} A. M. Bikchentaev and M. S. Moslehian, \textit{Characterizations of tracial functionals on matrix and operator algebras}, Bull. Belg. Math. Soc. Simon Stevin \textbf{32} (2025), no. 3, 343--361. 
	
	\bibitem{BM2} A. M. Bikchentaev and M. S. Moslehian, \textit{Characterizations of the canonical trace on full matrix algebras}, J. Math. Anal. Appl. \textbf{552} (2025), 129764.
	
	\bibitem{DinhDuNguyenVuong2023} T. H. Dinh, H. B. T. Du, A. N. D. Nguyen, and T. D. Vuong, \textit{On new quantum divergences}, Linear Multilinear Algebra. \textbf{71} (2023), 1--15.
	
	\bibitem{Dinh2013b} T. H. Dinh, M. T. Ho, and H. Osaka, \textit{A generalization of Powers-Størmer’s inequality }, Linear Algebra Appl. \textbf{439} (2013), 3035--3042.
	
	\bibitem{DinhLeOsakaPhan2025} T. H. Dinh, A. V. Le, H. Osaka, and N. Y. Phan, \textit{New quantum divergences generated by monotonicity inequality}, Math. Inequal. Appl. \textbf{28} (2025), 143--157.
	
	\bibitem{DinhTikhonov2010b} T. H. Dinh and O. E. Tikhonov, \textit{To the theory of operator monotone and operator convex functions}, Russian Math. \textbf{54} (2010), 7--11.
	
	\bibitem{FiedlerPtak1997} M. Fiedler and V. Pt\'{a}k, \textit{A new positive definite geometric mean of two positive definite matrices}, Linear Algebra Appl. \textbf{251} (1997), 1--20.
	
	\bibitem{Hayashi2017} M. Hayashi, \textit{Quantum Information Theory}, Springer, 2nd edition, 2017. 
	
	\bibitem{Jeong2024} M. Jeong, S. Kim, and T.-Y. Tam, \textit{New weighted spectral geometric mean and quantum divergence}, Linear Algebra. Appl. \textbf{726} (2025), 164--179.
	
	\bibitem{Jozsa1994} R. Jozsa, \textit{Fidelity for mixed quantum states}, J. Modern Opt. \textbf{41} (1994), 2315--2323.
	
	\bibitem{KMN} M. Kian, M. S. Moslehian, and R. Nakamoto, \textit{Asymmetric Choi--Davis inequalities}, Linear Multilinear Algebra. \textbf{70} (2022), no. 17, 3287--3300. 
	
	\bibitem{KuboAndo1980} F. Kubo and T. Ando, \textit{Means of positive linear operators}, Math. Ann. \textbf{246} (1980), 205--224.
	
	\bibitem{MO} M. S. Moslehian and H. Osaka, \textit{Advanced Techniques with Block Matrices of Operators}, Frontiers in Mathematics, Birkh\"{a}user, 2024.
	
	\bibitem{spehner1}	D. Spehner and M. Orszag, \textit{Geometric quantum discord with Bures distance}, New J. Phys. \textbf{15} (2013), no. 10, Article ID 103001, 18 p.
	
	\bibitem{Uhlmann1976} A. Uhlmann, \textit{The ``transition probability'' in the state space of a $*$-algebra}, Rep. Math. Phys. \textbf{9} (1976), 273--279.
	
\end{thebibliography}
\end{document}